\documentclass[journal]{IEEEtran}
\pdfoutput=1
\usepackage{amsmath,amssymb,amsfonts}
\usepackage{amsthm}
\usepackage{algorithmic}
\usepackage{algorithm}
\usepackage{array}
\usepackage{bbding}
\usepackage{booktabs}
\usepackage{color}
\usepackage{cite}
\usepackage{pifont}
\usepackage{graphicx}
\usepackage{multirow}
\usepackage[caption=false,font=scriptsize,labelfont=sf,textfont=sf]{subfig}
\usepackage{textcomp}
\usepackage{threeparttable}
\usepackage{stfloats}
\usepackage{url}
\usepackage{verbatim}
\usepackage{xcolor}

\begin{document}

\title{Opinion Dynamics in Two-Step Process: Message Sources, Opinion Leaders and Normal Agents}

\author{Huisheng Wang, Yuejiang Li, Yiqing Lin 
and H. Vicky Zhao,~\IEEEmembership{Senior Member,~IEEE}%
\thanks{Huisheng Wang, Yuejiang Li, Yiqing Lin and H. Vicky Zhao are with the Department of Automation, Tsinghua University, Beijing 100084, China, (e-mail: whs22@mails.tsinghua.edu.cn; lyj18@mails.tsinghua.edu.cn; linyq20@mails.tsinghua.edu.cn; vzhao@tsinghua.edu.cn).}%
}

\markboth{Journal of \LaTeX\ Class Files,~Vol.~14, No.~8, August~2022}%
{Wang \MakeLowercase{\textit{et al.}}: Opinion Dynamics in Two-step Social Networks: Sources, Opinion Leaders and Normal Agents.}

\IEEEpubid{0000--0000/00\$00.00~\copyright~2022 IEEE}

\maketitle

\begin{abstract}
According to mass media theory, the dissemination of messages and the evolution of opinions in social networks follow a two-step process. First, opinion leaders receive the message from the message sources, and then they transmit their opinions to normal agents. However, most opinion models only consider the evolution of opinions within a single network, which fails to capture the two-step process accurately. To address this limitation, we propose a unified framework called the Two-Step Model, which analyzes the communication process among message sources, opinion leaders, and normal agents.
In this study, we examine the steady-state opinions and stability of the Two-Step Model. Our findings reveal that several factors, such as message distribution, initial opinion, level of stubbornness, and preference coefficient, influence the sample mean and variance of steady-state opinions. Notably, normal agents' opinions tend to be influenced by opinion leaders in the two-step process.
We also conduct numerical and social experiments to validate the accuracy of the Two-Step Model, which outperforms other models on average. Our results provide valuable insights into the factors that shape social opinions and can guide the development of effective strategies for opinion guidance in social networks.
\end{abstract}

\begin{IEEEkeywords}
Message source, opinion model, opinion leader, two-step process, selective exposure.
\end{IEEEkeywords}

\section{Introduction}

\IEEEPARstart{T}{he} rapid growth of information technology and social media has allowed people to express their opinions freely, but uncontrolled expression can lead to social instability \cite{akram2017study}. Therefore, guiding public opinion has become a critical area of research \cite{monroe1998public, burstein2003impact}. To achieve this, we first need to understand the fundamental principles of opinion evolution. Researchers from various disciplines, including communication, psychology, and political science, are studying how public opinion forms and evolves \cite{wu2008public, leeper2014political, abebe2018opinion, chen2018quantifying}. 

To study how public opinion is shaped and influenced, it is essential to understand the communication process, as communication serves as the primary means by which individuals and groups share their opinions, beliefs, and attitudes with others \cite{mcquail2010mcquail}. In mass media theory, communication refers to the process of transmitting messages and opinions. Messages are content or information conveyed through media channels, while opinions represent individuals' personal beliefs or judgments on particular issues. There exists a close connection between messages and opinions, as messages aim to influence or shape public opinions. 

The mass media theory identifies three key participants in the communication process: message sources, opinion leaders, and normal agents. Message sources are entities, such as media organizations, journalists, or other relevant individuals or institutions, that create and distribute messages \cite{schram1982men}. Opinion leaders are individuals who have a significant influence on shaping the opinions of others, for example, experts in specific fields and influential figures in society \cite{lazarsfeld1968people}. Normal agents are individuals who receive and are influenced by the opinions expressed by opinion leaders \cite{mcquail1987mass}. 

The two-step process is one of the essential models of the communication process in mass media theory, which highlights the hierarchical flow of messages and opinions. In the first step, the messages flow from message sources to opinion leaders. In the second step, opinion leaders receive the messages first, form their opinions, and then transmit their opinions to the normal agents \cite{lazarsfeld1968people}. 
The two-step process indicates that the propagation of messages and opinions follows a multi-layer network structure \cite{dickison2016multilayer}. For example, in online social networks, Internet celebrities act as opinion leaders, and billions of regular users are normal agents. Regular users tend to adopt the opinions of Internet celebrities, whose messages are frequently retweeted, thus forming a two-step process \cite{choi2015two, jiang2021study, xie2021detecting}. These retweeted messages not only increase the visibility of Internet celebrities' opinions but also create a sense of validation and social proof among regular users, influencing their opinions and 
attitudes towards various topics.

\IEEEpubidadjcol

In the field of mass media studies, some qualitative models have been proposed to describe the communication process, such as the contagion model \cite{rogers2010diffusion}, multidimensional scaling model \cite{katz1959preliminary}, and complex network model \cite{morris2018internet}. They provide valuable insights into the dynamics of communication and are instrumental in shaping our understanding of how messages and opinions are transmitted and received. However, these qualitative models may not be sufficient to capture the quantitative properties of public opinion evolution. Quantitative models provide a framework to analyze and predict the dynamics of opinions and offer a more precise understanding of how opinions spread and evolve over time. Therefore, researchers have turned to approaches that integrate mathematical tools, known as opinion dynamics. 

Opinion dynamics models provide quantitative tools to analyze the evolution of opinions over networks, and 
has attracted lots of attention from researchers in mathematics, political science, economics, and social psychology \cite{kaur2013human, zhao2016opinion, carlson2019through, urena2019review, zha2020opinion}. 
Over the past years, numerous opinion dynamics models have been proposed, with the DeGroot (DG) model  \cite{anderson2019recent} and the Friedkin-Johnsen (FJ) model  \cite{noorazar2020recent} serving as the foundation. The DG model is a weighted averaging model \cite{degroot1974reaching}. It assumes that agents weigh the opinions of others based on their perceived expertise and update their opinions accordingly. Based on the DG model, the FJ model introduces agents' intrinsic beliefs to model the social behavior of stubborn agents who are difficult to influence \cite{friedkin1990social}. 
The DG and the FJ models are linear models that enable comprehensive theoretical analyses of the steady states of the opinion evolution process. Nevertheless, the simplicity of their model assumptions hinders their ability to capture the intricacies of the communication process in mass media theory, as demonstrated in the following.

First, the mass media theory suggests that the opinions of opinion leaders and normal agents follow a two-step process, meaning that opinion leaders and normal agents are situated in separate networks. This suggests that opinion exchange and influence occur among both opinion leaders and normal agents. However, the flow of opinions is unidirectional from opinion leaders to normal agents, meaning that opinion leaders influence the opinions of normal agents, but not the other way around. On the other hand, both the DG and FJ models allow for mutual influence between any two agents, disregarding the distinction between opinion leaders and normal agents. Instead, these models treat all agents as the same type. In other words, they do not differentiate between opinion leaders and normal agents, opting to model them all as the same type. Previous models on opinion leaders (OL) in opinion dynamics have demonstrated the influence of opinion leaders on the mean and variance of normal agents' opinions \cite{chen2016characteristics, lu2019fuzzy, shi2021leader}. They focus on the impact of opinion leaders on normal agents but do not form a unified framework to analyze the evolution patterns of opinions among opinion leaders and normal agents. $\square$

Second, message sources are considered the starting point of the two-step process. However, neither the DG nor the FJ models explicitly incorporate the message sources. To the best of our knowledge, there have been relatively few models that focus on modeling message sources. The Cyber-Social Network (CSN) model \cite{mao2018spread, mao2022social} is a pioneering work in this area, which incorporates both agents and message sources. This model effectively incorporates the concept of message sources into the communication process, but it does not consider the influence of opinion leaders. In this paper, we aim to build upon the CSN model and extend it to the two-step process.

Third, according to mass media theory, the selective exposure feature refers to agents' tendency to have biases and preferences when receiving messages and opinions \cite{sears1967selective, allahverdyan2014opinion}. In the DG and FJ models, an agent's weight coefficients towards others' opinions remain fixed and independent of the opinions. To account for selective exposure, researchers have modified the opinion-updating equation and proposed a class of bounded confidence models. The Hegselmann-Krause (HK) model \cite{hegselmann2002opinion} defines neighbor agents whose opinions are within a given range and follow the FJ model, while non-neighbor agents do not affect each other; the biased opinion formation (BOF) model \cite{dandekar2013biased} introduces a biased opinion term to model selective exposure; and the stochastic bounded confidence (SBC) model \cite{liu2013social, baumann2021emergence} defines the neighborhood probability of the agents as a function of their opinions' distance and updates the opinion according to the HK model. However, obtaining analytical solutions for steady-state (SS) opinions is challenging for these models. To address this issue, we introduce the concept of message preference, which captures selective exposure and provides approximate analytical solutions for steady-state opinions in the two-step process.

\begin{table}
\caption{Comparison between our model and existing models.}
\label{tab:tab0}
\centering
\setlength{\tabcolsep}{3.6pt}
\begin{tabular}{c|ccccccc|c}
\toprule
Model&DG&FJ&HK&BOF&SBC&OL&CSN&Our model\\
\midrule
Message source&\color{purple}{\ding{55}}&\color{purple}{\ding{55}}&\color{purple}{\ding{55}}&\color{purple}{\ding{55}}&\color{purple}{\ding{55}}&\color{purple}{\ding{55}}&\color{teal}{\ding{51}}&\color{teal}{\ding{51}}\\
Opinion leader&\color{purple}{\ding{55}}&\color{purple}{\ding{55}}&\color{purple}{\ding{55}}&\color{purple}{\ding{55}}&\color{purple}{\ding{55}}&\color{teal}{\ding{51}}&\color{purple}{\ding{55}}&\color{teal}{\ding{51}}\\
Normal agent&\color{teal}{\ding{51}}&\color{teal}{\ding{51}}&\color{teal}{\ding{51}}&\color{teal}{\ding{51}}&\color{teal}{\ding{51}}&\color{teal}{\ding{51}}&\color{teal}{\ding{51}}&\color{teal}{\ding{51}}\\
Selective exposure&\color{purple}{\ding{55}}&\color{purple}{\ding{55}}&\color{teal}{\ding{51}}&\color{teal}{\ding{51}}&\color{teal}{\ding{51}}&\color{teal}{\ding{51}}&\color{teal}{\ding{51}}&\color{teal}{\ding{51}}\\
\bottomrule
\end{tabular}
\end{table}

Different from most prior work that only focuses on one or two specific aspects of the mass media theory and has not formed a framework to describe the communication process, in this paper, we propose a unified Two-Step Model that jointly considers message sources, opinion leaders, and the normal agents, as well as the selective exposure feature (Table \ref{tab:tab0}). It should be noted that there are many forms of opinion models that we cannot list one by one, but to the best of our knowledge, there is currently no model proposed that presents a unified framework. Our investigation is crucial in shedding light on the communication process and its impact on the evolution of opinions in social networks.

Our contributions can be summarized as follows:
\begin{enumerate}
    \item Based on mass media theory and opinion dynamics, we propose a unified Two-Step Model, which includes message sources, opinion leaders, and normal agents.
    \item We introduce the mechanism of message preference to describe selective exposure for mathematical tractability. 
    \item We use dynamic system theory to analyze the properties of the model theoretically, including steady-state opinions and the influence of parameters.
    \item We conduct numerical and social experiments to verify the validity of our proposed model. 
\end{enumerate}

The rest of this paper is organized as follows. We introduce the model assumptions in Section \ref{sec:preliminary}, and propose the Two-Step Model in Section \ref{sec:basic}. We delve into the properties of the model and conduct the numerical and social experiments in Sections \ref{sec:analysis}, \ref{sec:experiments}, and \ref{sec:social}, respectively. Finally, we provide concluding remarks in Section \ref{sec:conclusion}.

\section{Model Assumptions}
\label{sec:preliminary}


In the two-step process, the message sources are the starting point and convey the messages to the opinion leaders. The opinion leaders get the messages from the sources, convert them into their opinions according to the selective exposure feature, and then send them to the normal agents. The normal agents access the opinions of the opinion leaders, interact with their neighbors and then form opinions. To simplify the modeling, we begin by making the following assumptions:

\begin{enumerate}
    \item Message sources' messages are mutually independent and correspond to a time-invariant distribution \cite{giardini2021opinion, mao2022social}.
    \item Opinion leaders have access to sources' messages \cite{lazarsfeld1968people}.
    \item Opinion leaders are not influenced by the opinions of any other opinion leaders, and their network is a null graph with no edges. \item Normal agents do not have access to sources' messages.
    \item Normal agents receive opinions from opinion leaders \cite{lazarsfeld1968people}, and the network between opinion leaders and normal agents is a complete bipartite graph. 
    \item Normal agents influence each other's opinion, and they form a complete graph network \cite{anderson2019recent}.
\end{enumerate}

It is worth noting that we consider the above simple scenario for mathematical tractability and to obtain insights. And we will take into consideration in our future work the more complicated circumstances where messages change over time, opinion leaders influence each other, and normal agents can access partial messages from sources.

\section{The Two-Step Model}
\label{sec:basic}

\begin{figure}[!t]
\centerline{\includegraphics[width=\columnwidth]{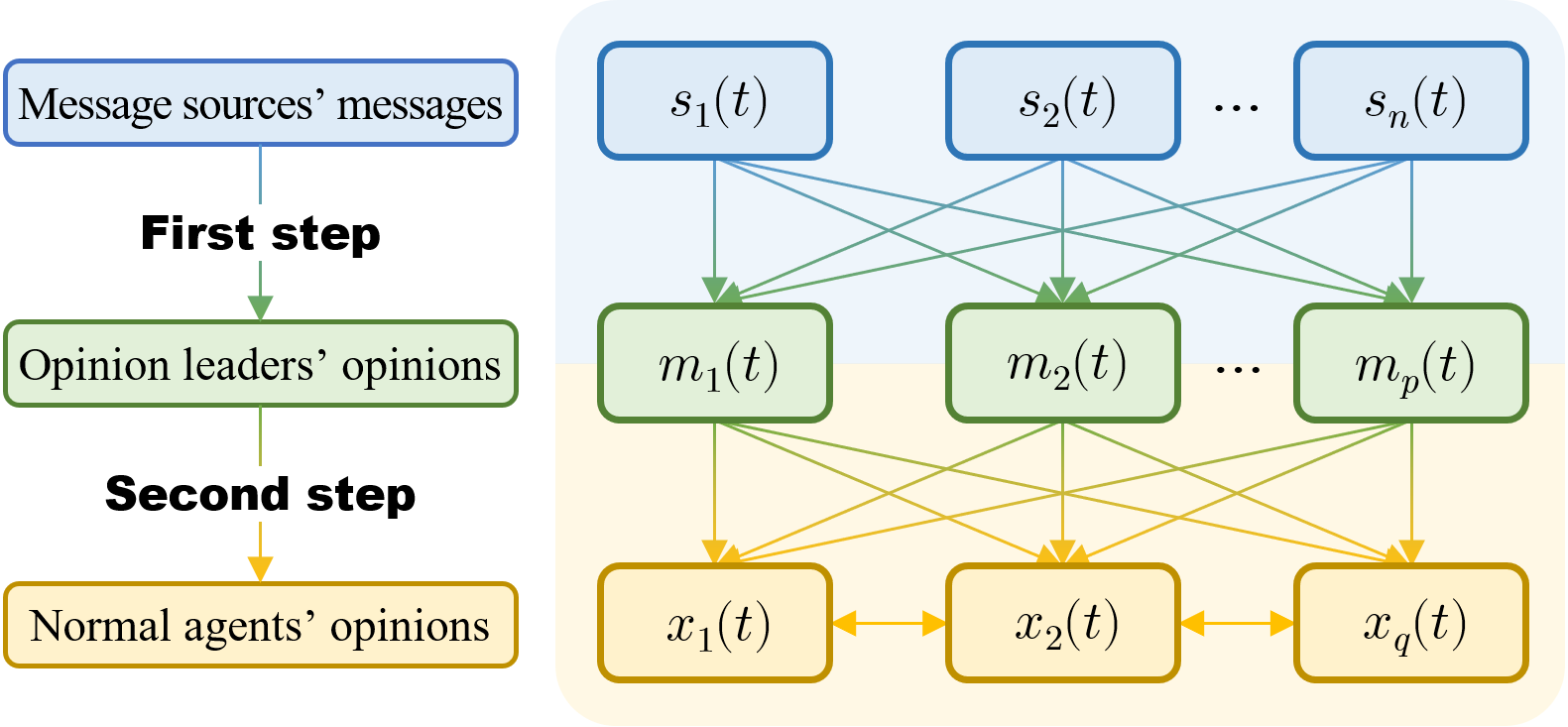}}
\caption{The communication process of the Two-Step Model.}
\label{fig:fig1}
\end{figure}

The Two-Step Model consists of message sources, opinion leaders, and normal agents. Consider a community of $n$ message sources, $p$ opinion leaders, and $q$ normal agents. Let $\boldsymbol{s}^{(t)}$ be the column vector of the sources' messages at time $t$. Let $\boldsymbol{m}^{(t)}$ and $\boldsymbol{x}^{(t)}$ be the column vector of the opinion leaders and normal agents' opinions at time $t$. Fig. \ref{fig:fig1} shows the communication process of the Two-Step Model.

\subsection{Message Source}
Following prior work \cite{mirtabatabaei2014eulerian, mao2018spread, mao2022social}, the sources' messages at time $t$, denoted by $\boldsymbol{s}^{(t)}$, is quantified by a series of random variable ranging in $[0,1]$. 

\newtheorem{definition}{\bf Definition}
\newtheorem{corollary}{\bf Corollary}
\newtheorem{theorem}{\bf Theorem}

\begin{definition}[Message distribution]
The message of source $i$ at time $t$ is a random variable that follows the time-invariant message distribution $\mathcal{S}$ independently and identically, i.e., 
\begin{equation}
    s_i^{(t)}\sim \mathcal{S},\forall i=1,2,\dots,n,t=1,2,\dots
    \label{equ:source}
\end{equation}
\end{definition}

The message $s_i^{(t)}$ close to $1$ represents that the source $i$ tends to support an event at time $t$ and vice versa.

\subsection{Opinion Leader}
At time $t$, each opinion leader receives all sources' message $\boldsymbol{s}^{(t)}$ without any distortion. Following the CSN model \cite{mao2018spread}, opinion leader $i$'s opinion, denoted by $m_i^{(t)}$, is a linear combination of the initial opinion $m_i^{(0)}$ and the average of the messages $\boldsymbol{s}^{(t)}$. The time-invariant weight coefficient $\sigma_i\in[0,1]$ is called the level of stubbornness. For a fully stubborn opinion leader $\sigma_i=1$, for a partially-stubborn opinion leader $0<\sigma_i<1$, and for a non-stubborn opinion leader $\sigma_i=0$. 
\begin{definition}[Opinion leaders' opinions]
The opinion of the opinion leader $i$ at time $t$ is given by
\begin{equation}
    m_i^{(t)}=\sigma_i m_i^{(0)}+(1-\sigma_i)\sum_{j=1}^n\gamma_{ij}^{(t)}s_j^{(t)}.
    \label{equ:m_i^{(t)}}
\end{equation}
\end{definition}

As shown in \eqref{equ:m_i^{(t)}}, $\gamma_{ij}^{(t)}$ is the weight of message $s_j^{(t)}$, and $\sum_{j=1}^n{\gamma}_{ij}^{(t)}=1$. And the vector $\boldsymbol{\gamma}_i^{(t)}=[\gamma_{i1}^{(t)},\gamma_{i2}^{(t)},\dots,\gamma_{in}^{(t)}]^\top$ is called the selective coefficient vector. 
Taking into account the selective exposure feature, if the distance between the opinion leader's opinion $m_i^{(t-1)}$ and the message $s_j^{(t)}$ is smaller, the selective coefficient $\gamma_{ij}^{(t)}$ should be larger. Before providing the calculation formula for the selective coefficient, we first introduce an auxiliary variable called message preference.

\begin{definition}[Message Preference]
The message preference of opinion leader $i$ for message source $j$ at time $t$ is given by
\begin{equation}
    p_{ij}^{(t)}=\{[m_i^{(t-1)}-s_j^{(t)}]^2\}^{\alpha-1}\{1-[m_i^{(t-1)}-s_j^{(t)}]^2\}^{\beta-1},
    \label{equ:message preference}
\end{equation}
where exponential parameters $\alpha\leqslant1$ and $\beta>1$ are called the preference coefficients.
\end{definition}
Equation \eqref{equ:message preference} shows that when $\alpha\leqslant1$ and $\beta>1$, opinion leaders prefer the messages closer to their own opinions at the previous instant, which is consistent with the selective exposure feature. And the preference coefficients $\alpha$ and $\beta$ reflect its strength: the larger value of $\alpha$ and the smaller value of $\beta$ means the stronger selective exposure feature. 
Because the selective coefficient $\boldsymbol{\gamma}_i^{(t)}$ is a weight vector, we define it as the normalized message preference $\boldsymbol{p}_i^{(t)}=[p_{i1}^{(t)},p_{i2}^{(t)},\dots,p_{in}^{(t)}]^\top$, i.e.,
\begin{equation*}
    \boldsymbol{\gamma}_i^{(t)}=\Vert\boldsymbol{p}_i^{(t)}\Vert_1^{-1}\boldsymbol{p}_i^{(t)}.
    \label{equ:selective-coefficent}
\end{equation*}

It should be noted that when $m_i^{(t-1)}=s_j^{(t)}$ and $\alpha<1$, the message preference $p_{ij}^{(t)}$ tends to infinity, and thus the selective coefficient $\gamma_{ij}^{(t)}=1$. That is to say, if the opinion leader receives a message that matches the opinion, the opinion will remain unchanged, which is consistent with \cite{10.2307/2118364}.

So far, we have defined the evolution equation of opinion leaders. The opinion of an opinion leader is the stubbornness-weighted average of the initial opinion and the weighted average of the messages they receive, and the weight of messages is determined by the normalized message preference that characterizes the selective exposure feature. We can rewrite \eqref{equ:m_i^{(t)}} as 
\begin{equation}
    \boldsymbol{m}^{(t)}=\boldsymbol{\mit\Sigma}\boldsymbol{m}^{(0)}+(\boldsymbol{I}-\boldsymbol{\mit\Sigma})\boldsymbol{\mit\Gamma}^{(t)}\boldsymbol{s}^{(t)},
    \label{equ:boldm}
\end{equation}
where $\boldsymbol{\mit\Gamma}^{(t)}=[\boldsymbol{\gamma}_1^{(t)},\boldsymbol{\gamma}_2^{(t)},\dots,\boldsymbol{\gamma}_p^{(t)}]^\top$ is the external influence matrix, which is row-stochastic, i.e., its entries are nonnegative and the sum of each row is $1$, 
and $\boldsymbol{\mit\Sigma}=\mathrm{diag}\{\sigma_1,\sigma_2,\dots,\sigma_p\}$ is the stubbornness matrix. 
For convenience, equation \eqref{equ:boldm} is named the message preference (MP) model, hereinafter.

\subsection{Normal Agent}
Following prior work \cite{friedkin1990social}, each normal agent's opinion is a linear function of all agents' opinions, including the opinion leaders and the normal agents. Specifically, at time $t$, normal agent $i$'s opinion, denoted by $x_i^{(t)}$, is a linear combination of the initial opinion $x_i^{(0)}$, the weighted average of the normal agents' opinions $\boldsymbol{x}^{(t-1)}$, and the weighted average of the opinion leaders' opinions $\boldsymbol{m}^{(t)}$. The corresponding weight coefficients are called the level of stubbornness $\rho_i\in[0,1]$, the normal agent influence proportion $\pi_i\in[0,1]$, and the opinion leader influence proportion $\theta_i\in[0,1]$, respectively, and $\rho_i+\pi_i+\theta_i=1$. 

Suppose that the network between normal agents is a complete graph, and its row-stochastic influence matrix, denoted by $\boldsymbol{W}$, is called the normal agent influence matrix, whose entry $w_{ij}\in[0,1]$ is the level of influence of normal agent $j$ on $i$. The network between opinion leaders and normal agents is a complete bipartite graph, and its row-stochastic influence matrix, denoted by $\boldsymbol{U}$, is called the opinion leader influence matrix, whose entry $u_{ij}\in[0,1]$ is the level of influence of opinion leader $j$ on normal agent $i$.

\begin{definition}[Normal agents' opinions]
The opinion of the normal agent $i$ at time $t$ is a combination of her initial opinion $x_i^{(0)}$ is given by
\begin{equation}
    x_i^{(t)}=\rho_i x_i^{(0)}+\pi_i \sum_{j=1}^q{w_{ij}x_j^{(t-1)}}+\theta_i \sum_{j=1}^p{u_{ij}m_j^{(t)}}.
    \label{equ:x_j^{(t)}}
\end{equation}
\end{definition}
Equation \eqref{equ:x_j^{(t)}} can be rewritten with matrices and vectors as 
\begin{equation}
    \boldsymbol{x}^{(t)}=\boldsymbol{P}\boldsymbol{x}^{(0)}+\boldsymbol{\mit\Pi}\boldsymbol{W}\boldsymbol{x}^{(t-1)}+\boldsymbol{\mit\Theta}\boldsymbol{U}\boldsymbol{m}^{(t)},
    \label{equ:boldx}
\end{equation}
where $\boldsymbol{P}=\mathrm{diag}\{\rho_1,\rho_2,\dots,\rho_q\}$, $\boldsymbol{\mit\Pi}=\mathrm{diag}\{\pi_1,\pi_2,\dots,\pi_q\}$, and $\boldsymbol{\mit\Theta}=\mathrm{diag}\{\theta_1,\theta_2,\dots,\theta_q\}$ are called the stubbornness matrix, the normal agent influence proportion matrix, and the opinion leader influence proportion matrix, respectively.


\section{Analysis of the Opinion Dynamics}
\label{sec:analysis}
We have proposed the Two-Step Model and presented the message and opinion dynamics for message sources, opinion leaders, and normal agents, as described in \eqref{equ:source}, \eqref{equ:boldm}, and \eqref{equ:boldx}, respectively. In this section, we aim to determine the steady-state opinions of \eqref{equ:boldm} and \eqref{equ:boldx}, and analyze how the model parameters affect the steady-state opinions. 

\subsection{Steady-state Analysis}
\subsubsection{Opinion leader's opinion}
We begin with a degenerate case $\alpha=\beta=1$, where opinion leaders have the same preference for all messages. In this case, equation \eqref{equ:m_i^{(t)}} is simplified as
\begin{equation*}
    m_i^{(t)}=\sigma_im_i^{(0)}+(1-\sigma_i)\frac{1}{n}\sum_{j=1}^ns_j^{(t)}.
\end{equation*}
When the number of sources $n$ tends to infinity, the opinion is almost surely equal to the weighted average of the initial opinion and the expectation of message distribution according to the Central Limit Theorem (CLT). Therefore, the opinion is time-invariant, and the steady-state opinion is
\begin{equation}
    \tilde{m}_i=m_i^{(t)}=\sigma_im_i^{(0)}+(1-\sigma_i)\mu,
    \label{equ:tildem-alpha=beta=1}
\end{equation}
when $\alpha=\beta=1$. When $\alpha\leqslant1$ and $\beta>1$, equation \eqref{equ:boldm} is a nonlinear difference equation, and it is expensive to find a closed-form solution. To overcome this issue, we adopt a data-fitting method to approximate the steady-state opinion.

\begin{theorem}[Steady-state opinion leaders' opinions]
If the number of sources $n$ is large sufficiently, the numerical steady-state opinions of equation \eqref{equ:boldm} is approximately given by
\begin{equation}
    \tilde{\boldsymbol{m}}\approx\boldsymbol{\mit\Sigma}^{\frac{\lambda\ln{\alpha}+1}{\kappa(\beta-1)+1}} \boldsymbol{m}^{(0)}+\mu\left(\boldsymbol{I}-\boldsymbol{\mit\Sigma}^{\frac{\lambda \ln{\alpha}+1}{\kappa(\beta-1)+1}} \right)\boldsymbol{1},
    \label{equ:m}
\end{equation}
when $\alpha\leqslant1$ and $\beta>1$, where $\mu$ is the expectation of the message distribution, and $\lambda\approx1.15$ and $\kappa\approx0.18$ are the constant parameters independent of $\alpha$, $\beta$, and $\boldsymbol{\mit\Sigma}$, etc. 
\label{the:the1}
\end{theorem}
\begin{proof}
First, we consider the conditions of $\alpha=1$ and $\beta>1$. In this case, equation \eqref{equ:m_i^{(t)}} can be derived as
\begin{equation*}
    m_i^{(t)}=\sigma_im_i^{(0)}+(1-\sigma_i)\frac{\sum_{j=1}^n{\{1-[m_i^{(t-1)}-s_j^{(t)}]^{2}\}^{\beta-1}s_j^{(t)}}}{\sum_{j=1}^n{\{1-[m_i^{(t-1)}-s_j^{(t)}]^{2}\}^{\beta-1}}}.
\end{equation*}
When $n$ tends to infinity, the summation is almost surely equal to the expectation according to the CLT. And the steady-state equation is given by
\begin{equation}
    \tilde{m}_i=\sigma_im_i^{(0)}+(1-\sigma_i)\frac{\mathbb{E}_{s\sim\mathcal{S}}\{[1-(\tilde{m}_i-s)^{2}]^{\beta-1}s\}}{\mathbb{E}_{s\sim\mathcal{S}}\{{[1-(\tilde{m}_i-s)^{2}]^{\beta-1}\}}}.
    \label{equ:tildem-explicit}
\end{equation}
Equation \eqref{equ:tildem-explicit} is a nonlinear equation for $\tilde{m}_i$, which can be solved by numerical methods. 
Inspired by equation \eqref{equ:tildem-alpha=beta=1}, we approximate the steady-state opinion $\tilde{m}_i$ as a weighted average of the initial opinion and the expectation of message distribution. The weight coefficient is a function of $\sigma_i$ and $\beta$, and we observe that it can be approximated by an exponential function of $\sigma_i$ with an exponent given by $w_\beta$. Specifically, we assume that
\begin{equation*}
    \tilde{m}_i\approx \sigma_i^{w_\beta}m_i^{(0)}+\left(1-\sigma_i^{w_\beta}\right)\mu.
    \label{equ:tildem-approx}
\end{equation*}
To determine the weight coefficient, we use the numerical solution from \eqref{equ:tildem-explicit}. By fitting the curve of $\beta$ and $w_\beta$, we find that $w_\beta$ is a fractional function $w_\beta=[\kappa(\beta-1)+1]^{-1}$, where $\kappa\approx0.18$ is a constant parameter. Therefore,
\begin{equation*}
    \tilde{m}_i\approx\sigma_i^{\frac{1}{\kappa(\beta-1)+1}}m_i^{(0)}+\left[1-\sigma_i^{\frac{1}{\kappa(\beta-1)+1}}\right]\mu.
\end{equation*}
Notably, this is a specific form of \eqref{equ:m} when $\alpha=1$ and $\beta>1$.

Next, we consider the case when $\alpha<1$ and $\beta>1$. Similarly, we assume that the steady-state opinion can be approximately expressed as
\begin{equation*}
    \tilde{m}_i\approx\sigma_i^{\frac{w_\alpha}{\kappa(\beta-1)+1}}m_i^{(0)}+\left[1-\sigma_i^{\frac{w_\alpha}{\kappa(\beta-1)+1}}\right]\mu,
\end{equation*}
where $w_\alpha$ is a function of $\alpha$. By fitting the curve of $\alpha$ and $w_\alpha$, we find that $w_\alpha$ is a logarithmic function $w_\alpha=\lambda\ln{\alpha}+1$, where $\lambda\approx1.15$ is a constant parameter. Therefore,
\begin{equation*}
    \tilde{m}_i\approx\sigma_i^{\frac{\lambda\ln{\alpha}+1}{\kappa(\beta-1)+1}}m_i^{(0)}+\left[1-\sigma_i^{\frac{\lambda\ln{\alpha}+1}{\kappa(\beta-1)+1}}\right]\mu.
\end{equation*}
So far we get the equation \eqref{equ:m}.
\end{proof}

It is worth noting that the ``sufficiently large'' condition for the number of sources is easily satisfied. As shown in Section \ref{sec:experiments}, equation \eqref{equ:m} provide a good approximation of the steady-state opinions when $n$ is greater than or equal to $3$. 

Theorem \ref{the:the1} indicates that the steady-state opinion is approximately equal to a weighted average of the initial opinion and the expectation of message distribution, and the weight coefficient is a modified level of stubbornness, which is a function of $\alpha$, $\beta$, and $\boldsymbol{\mit\Sigma}$. For convenience, we denote the modified level of stubbornness in the steady-state opinion of opinion leaders in \eqref{equ:m} as $\boldsymbol{Z}=\boldsymbol{\mit\Sigma}^{\frac{\lambda \ln{\alpha}+1}{\kappa(\beta-1)+1}}$.

Next, we examine the stability of steady-state opinions.

\begin{theorem}[Stability]\label{the:the2}
When $\alpha\in(0.5,1]$ and $\beta\in(1,+\infty)$, the steady-state opinion in \eqref{equ:m} is asymptotically stable. 
\end{theorem}
\begin{proof}
The proof is shown in the supplementary materials.
\end{proof}
Theorem \ref{the:the2} states that when $\alpha\in(0.5,1]$ and $\beta\in(1,+\infty)$, the steady-state opinions of opinion leaders and normal agents are stable. This implies that given initial opinions, the steady-state opinions can be approximately predicted by Theorem \ref{the:the1}.

\subsubsection{Normal agents' opinion}
The steady-state opinions and stability analysis is shown in Theorem \ref{the:the1-2}.
\begin{theorem}[Steady-state normal agents' opinions and stability]\label{the:the1-2}
The steady-state opinions of equation \eqref{equ:boldx} is given by
\begin{equation}
    \tilde{\boldsymbol{x}}=(\boldsymbol{I}-\boldsymbol{\mit\Pi}\boldsymbol{W})^{-1} [\boldsymbol{P}\boldsymbol{x}^{(0)}+\boldsymbol{\mit\Theta}\boldsymbol{U}\tilde{\boldsymbol{m}}].
    \label{equ:x}
\end{equation}
And the steady-state solution is asymptotically stable.
\end{theorem}
\begin{proof}
Equation \eqref{equ:boldx} is a linear difference equation. The steady-state equation is given by
\begin{equation*}
    \tilde{\boldsymbol{x}}=\boldsymbol{P}\boldsymbol{x}^{(0)}+\boldsymbol{\mit\Pi}\boldsymbol{W}\tilde{\boldsymbol{x}}+(\boldsymbol{I}-\boldsymbol{P}-\boldsymbol{\mit\Pi})\boldsymbol{U}\tilde{\boldsymbol{m}}.
\end{equation*}
And equation \eqref{equ:x} can be obtained by matrix operations. Following the prior work \cite{friedkin1990social}, we can prove the asymptotical stability of the steady-state solution.
\end{proof}


\subsection{Parameter Analysis}
Following prior work \cite{murray2004consensus, gionis2013opinion, liu2018polarizability, ye2019consensus, cisneros2020polarization, li2021evaluation}, we pay attention to the consensus and polarization attributes of the steady-state opinions, which can be quantified by the sample mean $\bar{\tilde{m}}=\frac{1}{p}\sum_{i=1}^p\tilde{m}_i$ and $\bar{\tilde{x}}=\frac{1}{q}\sum_{i=1}^q\tilde{x}_i$, and the sample variance $\sigma_{\tilde{m}}^2=\frac{1}{p}\sum_{i=1}^p(\tilde{m}_i-\bar{\tilde{m}})^2$ and $\sigma_{\tilde{x}}^2=\frac{1}{q}\sum_{i=1}^q(\tilde{x}_i-\bar{\tilde{x}})^2$, respectively.
To simplify the analysis, we assume that the stubbornness matrix and proportion matrix are scalar matrices, i.e., $\boldsymbol{\mit\Sigma}=\sigma\boldsymbol{I}$, $\boldsymbol{Z}=z\boldsymbol{I}$,  $\boldsymbol{P}=\rho\boldsymbol{I}$, $\boldsymbol{\mit\Pi}=\pi\boldsymbol{I}$, and $\boldsymbol{\mit\Theta}=\theta\boldsymbol{I}$, and the normal agent influence matrix and opinion leader influence matrix are the row-normalized all-one matrices, i.e., $\boldsymbol{W}=q^{-1}\boldsymbol{E}_{q,q}$ and $\boldsymbol{U}=q^{-1}\boldsymbol{E}_{p,q}$. First, we provide the sample mean and variance of steady-state opinions, which is given by Theorem \ref{lem:1}. 

\begin{theorem}[Sample mean and variance of steady-state opinions]\label{lem:1}
The sample mean and variance of the opinion leaders and normal agents' steady-state opinions are given by
\begin{align*}
    &\bar{\tilde{m}}\approx z\bar{m}^{(0)}+(1-z)\mu, &&\sigma^2_{\tilde{m}}\approx z^2\sigma^2_{m^{(0)}},\\
    &\bar{\tilde{x}}=\frac{\rho}{\rho+\theta}\bar{x}^{(0)}+\frac{\theta}{\rho+\theta}\bar{\tilde{m}}, &&\sigma^2_{\tilde{x}}=\rho^2\sigma^2_{x^{(0)}},
\end{align*}
where $\bar{m}^{(0)}$ and $\sigma^2_{m^{(0)}}$ are the sample mean and variance of the opinion leaders' initial opinions, and $\bar{x}^{(0)}$ and $\sigma^2_{x^{(0)}}$ are the sample mean and variance of normal agents.
\end{theorem}

\begin{proof}
1) Sample mean and variance of opinion leaders' steady-state opinions. According to Theorem \ref{the:the1}, opinion leader $i$'s steady-state opinion is approximately given by
\begin{equation*}
    \tilde{m}_i\approx zm_i^{(0)}+(1-z)\mu.
\end{equation*}
And the sample mean and variance of opinion leaders' steady-state opinions can be easily calculated.

2) Sample mean and variance of normal agents' steady-state opinions. When $\boldsymbol{\mit\Pi}=\pi\boldsymbol{I}$ and $\boldsymbol{W}=\frac{1}{q}\boldsymbol{E}_{q,q}$, we have $(\boldsymbol{I}-\pi\boldsymbol{W})^{-1}=\boldsymbol{I}+\frac{\pi}{1-\pi}\boldsymbol{W}$, because
\begin{align*}
    (\boldsymbol{I}-\pi\boldsymbol{W})\big(\boldsymbol{I}+\frac{\pi}{1-\pi}\boldsymbol{W}\big)&=\boldsymbol{I}+\frac{\pi^2}{1-\pi}\boldsymbol{W}-\frac{\pi^2}{1-\pi}\boldsymbol{W}^2\\
    &=\boldsymbol{I}+\frac{\pi^2}{1-\pi}\boldsymbol{W}-\frac{\pi^2}{1-\pi}\boldsymbol{W}=\boldsymbol{I}.
\end{align*}
According to Theorem \ref{the:the1-2}, the steady-state opinions can be expressed as
\begin{equation*}
    \tilde{\boldsymbol{x}}=\left(\boldsymbol{I}+\frac{\pi}{1-\pi}\boldsymbol{W}\right)[\rho \boldsymbol{x}^{(0)}+\theta\boldsymbol{U}\tilde{\boldsymbol{m}}].
\end{equation*}
Therefore, normal agent $i$'s steady-state opinion is given by
\begin{align*}
    \tilde{x}_i&=\rho x_i^{(0)}+\theta\bar{\tilde{m}}+\frac{\pi\rho}{1-\pi}\bar{x}^{(0)}+\frac{\pi\theta}{1-\pi}\bar{\tilde{m}}\\
    &=\rho x_i^{(0)}+\frac{(1-\rho-\theta)\rho}{\rho+\theta}\bar{x}^{(0)}+\frac{\theta}{\rho+\theta}\bar{\tilde{m}}.
\end{align*}
And the sample mean and variance of normal agents' steady-state opinions can be calculated.
\end{proof}

Next, we analyze how the steady-state opinions in \eqref{equ:boldm} and \eqref{equ:boldx} is affected by the model parameters including:
\begin{enumerate}
    \item Message distribution: $\mathcal{S}$.
    \item Initial opinion: $\boldsymbol{m}^{(0)}$ and $\boldsymbol{x}^{(0)}$.
    \item Level of stubbornness: $\boldsymbol{\mit\Sigma}$ and $\boldsymbol{P}$.
    \item Preference coefficient: $\alpha$ and $\beta$.
\end{enumerate}

Literature \cite{friedkin1990social} has analyzed the influence matrix's impact on steady-state opinion. And the analysis of the proportion matrix's impact is similar to that of the stubbornness matrix. Therefore, we omit the analysis of these parameters. According to Theorem \ref{lem:1}, we can analyze the influence of model parameters and obtain the following corollaries.

\begin{corollary}[Message distribution]
The sample mean of the opinion leaders and normal agents' steady-state opinions is approximately a nondecreasing function of the expectation of the message distribution.
\label{the:the3}
\end{corollary}

\begin{proof}
According to Theorem \ref{lem:1}, it is clearly established.
\end{proof}

Corollary \ref{the:the3} demonstrates a positive relationship between the sample mean of the steady-state opinions and the expectation of the message distribution. This finding has significant implications in mass media theory, as it suggests that the mean of the messages transmitted by the media can directly influence the steady-state opinions of both opinion leaders and normal agents.
This finding underscores the importance of understanding the role of the message sources in shaping public opinion and highlights the need for careful consideration of the messages transmitted by the sources in the context of opinion formation and evolution.


\begin{corollary}[Initial opinion]
The sample mean and variance of the opinion leaders and normal agents' steady-state opinions are nondecreasing linear functions of the sample mean and variance of initial opinions.
\label{the:the7}
\end{corollary}

\begin{proof}
According to Theorem \ref{lem:1}, it is clearly established.
\end{proof}

Corollary \ref{the:the7} demonstrates a positive relationship between the sample mean and variance of the steady-state opinions and the sample mean and variance of the initial opinions. From the perspective of mass media theory, initial opinions correspond to the intrinsic beliefs \cite{schram1982men, friedkin1990social}, of which the sample mean and variance promotes those of the steady-state opinion in the same direction.

\begin{corollary}[Level of stubbornness]~
\label{the:the5}
\begin{enumerate}
    \item The sample mean and variance of opinion leaders' steady-state opinions are approximately power functions of the level of stubbornness, i.e.,
    \begin{equation*}
        \bar{\tilde{m}}\approx[\bar{m}^{(0)}-\mu]\sigma^{\frac{\lambda \ln{\alpha}+1}{\kappa(\beta-1)+1}}+\mu,\sigma^2_{\tilde{m}}\approx\sigma^2_{m^{(0)}}\sigma^{\frac{2(\lambda \ln{\alpha}+1)}{\kappa(\beta-1)+1}}.
    \end{equation*}
    \item The sample mean and variance of normal agents' steady-state opinions are fractional function and quadratic function of the level of stubbornness, respectively, i.e.,
    \begin{equation*}
        \bar{\tilde{x}}=\frac{\rho\bar{x}^{(0)}+\theta\bar{\tilde{m}}}{\rho+\theta}, \sigma^2_{\tilde{x}}=\rho^2\sigma^2_{x^{(0)}}.
    \end{equation*}
\end{enumerate}
\end{corollary}

\begin{proof}
According to Theorem \ref{lem:1}, it is established.
\end{proof}

According to Corollary \ref{the:the5}, as the level of stubbornness increases, the sample mean of opinion leaders' steady-state opinions shifts from the expectation of message distribution to the sample mean of initial opinions. Likewise, the sample mean of normal agents' steady-state opinions shifts from the sample mean of opinion leaders' steady-state opinions to the sample mean of initial opinions. This suggests that the level of stubbornness influences the strength of intrinsic beliefs. 
Corollary \ref{the:the5} also shows that a greater level of stubbornness induces a greater sample variance of steady-state opinions. From the perspective of the mass media theory, when the level of stubbornness is high, agents become less likely to be influenced by other messages or opinions and more likely to maintain their intrinsic opinions. This leads to a broader range of steady-state opinions, as agents are less likely to adopt a common opinion.

\begin{corollary}[Preference coefficient]~
\begin{enumerate}
    \item When $\sigma=0$ or $\sigma=1$, the sample variance of opinion leaders' steady-state opinions is not affected by the preference coefficients.
    \item When $0<\sigma<1$, the sample variance of opinion leaders' steady-state opinions decreases with $\beta$ and increases with $\alpha$.
\end{enumerate}
\label{the:the6}
\end{corollary}

\begin{proof}
If $\sigma=0$ or $\sigma=1$, then $z=0$ or $z=1$, which is independent with $\alpha$ and $\beta$. According to Theorem \ref{lem:1}, the preference coefficients will not affect the sample variance of opinion leaders' steady-state opinions. 
When $0<\sigma<1$, the partial derivatives of $z$ over $\alpha$ and $\beta$ are
$\frac{\partial z}{\partial\alpha}=\frac{\lambda z\ln{\sigma}}{\alpha[\kappa(\beta-1)+1]}<0$ and
$\frac{\partial z}{\partial\beta}=-\frac{(\lambda\ln{\alpha}+1)\kappa z\ln{\sigma}}{[\kappa(\beta-1)+1]^2}>0$. Therefore, the larger $\alpha$ and the smaller $\beta$, the smaller $z$, and the smaller the sample variance of opinion leaders' steady-state opinions, according to Theorem \ref{lem:1}.
\end{proof}

Corollary \ref{the:the6} shows that increasing $\alpha$ or decreasing $\beta$ will reduce the sample variance of opinion leaders' steady-state opinions, which supports the selective exposure feature in the communication process. The larger $\alpha$ and smaller $\beta$, the greater the preference for receiving messages far from their intrinsic opinions, and the wider the horizon to get messages. Specifically, when $\alpha$ is large, opinion leaders are less likely to adhere to their intrinsic opinions; when $\beta$ is small, opinion leaders are more likely to expose themselves to messages that challenge their intrinsic opinions, leading to a wider horizon for receiving messages. This behavior, in turn, reduces the sample variance of steady-state opinions.

In summary, the sample mean of opinion leaders and normal agents' steady-state opinions $\bar{\tilde{m}}$ and $\bar{\tilde{x}}$ are influenced by the following three factors:
\begin{enumerate}
    \item \textbf{Message distribution}. The sample mean of steady-state opinions is an increasing linear function of the expectation of message distribution $\mu$.
    \item \textbf{Initial opinion}. The sample mean of steady-state opinions is an increasing linear function of the sample mean of initial opinions $\bar{m}^{(0)}$ and $\bar{x}^{(0)}$.
    \item \textbf{Level of stubbornness}. As level of stubbornness $\sigma$ and $\rho$ increase, the sample mean of steady-state opinions shifts from the expectation of message distribution $\mu$ or the sample mean of opinion leaders' steady-state opinions $\bar{\tilde{m}}$ to the sample mean of initial opinions $\bar{m}^{(0)}$ and $\bar{x}^{(0)}$. 
\end{enumerate}

The sample variance of steady-state opinions $\sigma^2_{\tilde{m}}$ and $\sigma^2_{\tilde{x}}$ are influenced by three factors as well:
\begin{enumerate}
    \item \textbf{Initial opinion}. The sample variance of steady-state opinions is an increasing linear function of the sample variance of initial opinions $\sigma^2_{m^{(0)}}$ and $\sigma^2_{x^{(0)}}$.
    \item \textbf{Level of stubbornness}. The sample variance of steady-state opinions is an increasing power or quadratic function of the sample variance of the level of stubbornness $\sigma$ and $\rho$.
    \item \textbf{Preference coefficient}. The sample variance of opinion leaders' steady-state opinions decreases with $\beta$ and increases with $\alpha$ when $0<\sigma<1$.
\end{enumerate}

\section{Numerical Experiments}
\label{sec:experiments}
In this section, we run simulations over synthetic networks to validate our Two-Step Model and discuss how opinions are affected by the model parameters.

The simulation network that we generate consists of $10^4$ message sources, $10^3$ opinion leaders, and $10^3$ normal agents. We create random matrices and normalize them row by row to generate the influence matrices $\boldsymbol{W}$ and $\boldsymbol{U}$. For opinion leaders, we generate a diagonal stubbornness matrix $\boldsymbol{\mit\Sigma}$, where the diagonal elements are from a uniform distribution on the interval $[0, 1]$. Similarly, we generate three random diagonal matrices for normal agents: the stubbornness matrix $\boldsymbol{P}$, the internal influence proportion matrix $\boldsymbol{\mit\Pi}$, and the external influence proportion matrix $\boldsymbol{C}$, and ensure that the sum of these three matrices equals the identity matrix.
We use the Beta distribution, denoted by $\mathcal{B}(a,b)$, to represent the message distribution, and we assume that the distributions of initial opinions of opinion leaders and normal agents are also Beta distributions, denoted by $\mathcal{B}(a_m,b_m)$ and $\mathcal{B}(a_x,b_x)$. The preference coefficients for opinion leaders are set to $\alpha\in(0.5,1]$ and $\beta\in(1,5)$.  Experimental results indicate that agents' opinions almost reach a steady state when $t\geqslant10^2$. Therefore, we choose $T=10^2$ as the simulation time length.

\subsection{Steady-state Analysis}

\begin{figure}[!t]
\centering
\subfloat[Correlation coefficient]{\includegraphics[width=0.5\columnwidth]{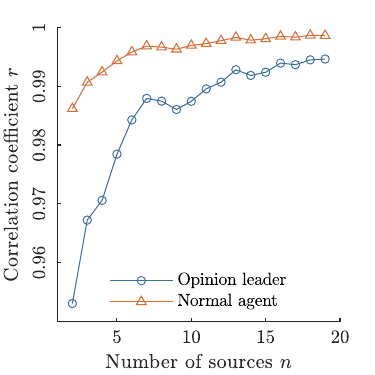} \label{fig:fig2a}}
\subfloat[Number of sources $n=10^3$]{\includegraphics[width=0.5\columnwidth]{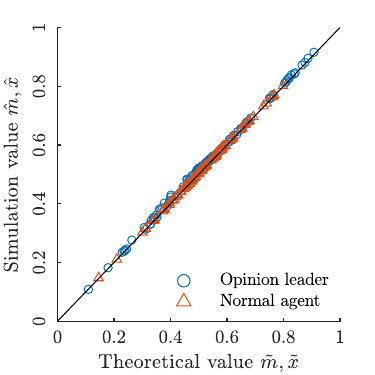} \label{fig:fig2b}}
\caption{Numerical experiment results: the steady-state analysis.}
\label{fig:fig2}
\end{figure}

We assume that the number of sources $n$ is sufficiently large in Theorem \ref{the:the1}. To test the correlation between the simulation values $\tilde{m}$, $\tilde{x}$, and theoretical values $\hat{m}$, $\hat{x}$, we vary $n$ from $2$ to $20$. As shown in Fig. \ref{fig:fig2a}, the correlation coefficient $r$ increases as $n$ increases. When $n=10^3$, the simulation and theoretical values are approximately distributed on the curve of the function $f(x)=x$, so the correlation coefficient $r$ tends to $1$, which is shown in Fig. \ref{fig:fig2b}. Even for $n\geqslant3$, the correlation coefficient $r$ for opinion leaders and normal agents is larger than $0.96$, indicating that Theorem \ref{the:the1} is approximately valid for $n\geqslant3$.

\subsection{Parameter Analysis}
Next, we separately investigate the effects of parameters on the sample mean and variance of steady-state opinions. Following the assumption in Section \ref{sec:analysis}, we select the scalar matrix as the stubbornness and proportion matrix, i.e., $\boldsymbol{\mit\Sigma}=\sigma\boldsymbol{I}$, $\boldsymbol{Z}=z\boldsymbol{I}$,  $\boldsymbol{P}=\rho\boldsymbol{I}$, $\boldsymbol{\mit\Pi}=\pi\boldsymbol{I}$, and $\boldsymbol{\mit\Theta}=\theta\boldsymbol{I}$, and the row-normalized all-one matrix as the influence matrix, i.e., $\boldsymbol{W}=q^{-1}\boldsymbol{E}_{q,q}$ and $\boldsymbol{U}=q^{-1}\boldsymbol{E}_{p,q}$. 
Below, we will use the method of controlling variables to study a particular model parameter, while keeping the other parameters constant. The default parameters we have selected are shown in Table \ref{tab:tab5}. 

\begin{table}
\caption{Default values of model parameters.}
\label{tab:tab5}
\centering
\setlength{\tabcolsep}{3pt}
\begin{tabular}{cc|cc|cc|cc}
\toprule
Parameter&Value&Parameter&Value&Parameter&Value&Parameter&Value\\
\midrule
$n$&10\textsuperscript{4}&$\sigma$&1/2&$\alpha$&1&$a_m$&1\\
$p$&10\textsuperscript{3}&$\rho$&1/3&$\beta$&2&$b_m$&1\\
$q$&10\textsuperscript{3}&$\pi$&1/3&$a$&1&$a_x$&1\\
$T$&10\textsuperscript{2}&$\theta$&1/3&$b$&1&$b_x$&1\\
\bottomrule
\end{tabular}
\end{table}

\subsubsection{Sample mean of steady-state opinions}

\begin{figure*}[!t]
\centering
\subfloat[Message distribution]{\includegraphics[width=1.75in]{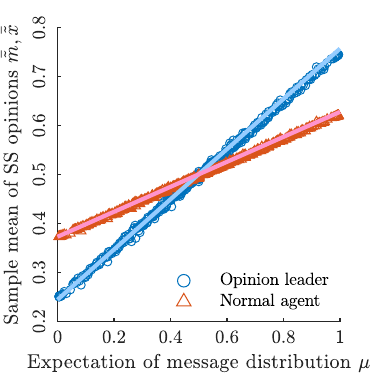} \label{fig:fig3a}}
\subfloat[Initial opinion]{\includegraphics[width=1.75in]{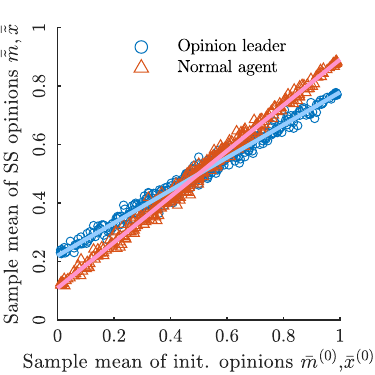} \label{fig:fig3b}}
\subfloat[Level of stubbornness: opinion leader]{\includegraphics[width=1.75in]{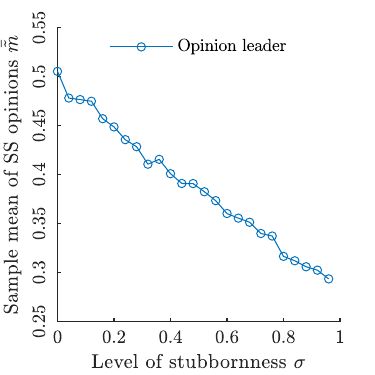} \label{fig:fig3c}}
\subfloat[Level of stubbornness: normal agent]{\includegraphics[width=1.75in]{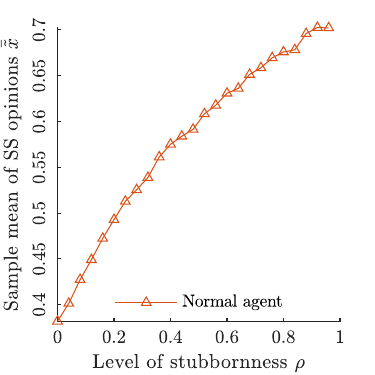} \label{fig:fig3d}}
\caption{Numerical experiment results: the sample mean of steady-state opinions.}
\label{fig:fig3}
\end{figure*}

Three factors influence the sample mean of steady-state opinions: message distribution, initial opinion, and level of stubbornness.

\textbf{Message distribution}. We take the parameters $(a,b)$ of the message distribution throughout the interval $(0,2)\times(0,2)$ to change the expectation of message distribution $\mu$, which is equal to $\frac{a}{a+b}$. The relationship between the sample mean of steady-state opinions $\bar{\tilde{m}}$, $\bar{\tilde{x}}$ and the expectation of message distribution $\mu$ is shown in Fig. \ref{fig:fig3a}. And we can see that ceteris paribus, the sample mean of steady-state opinions increases linearly with the expectation of message distribution, which is consistent with Corollary \ref{the:the3}.

\textbf{Initial opinion}. We take the parameters $(a_m,b_m)$ throughout the interval $(0,2)\times(0,2)$ to change the sample mean of opinion leaders' initial opinions $\bar{m}^{(0)}$, which is equal to $\frac{a_m}{a_m+b_m}$. And we take the parameters $(a_x,b_x)$ throughout the interval $(0,2)\times(0,2)$ to change the sample mean of normal agents' initial opinions $\bar{x}^{(0)}$, which is equal to $\frac{a_x}{a_x+b_x}$. The relationship between the sample mean of steady-state opinions $\bar{\tilde{m}}$, $\bar{\tilde{x}}$ and the sample mean of initial opinions $\bar{m}^{(0)}$ and $\bar{x}^{(0)}$ is shown in Fig. \ref{fig:fig3b}. And we can see that ceteris paribus, the sample mean of steady-state opinions increases linearly with the sample mean of initial opinions, which is consistent with Corollary \ref{the:the7}.

\textbf{Level of stubbornness}. According to Section Corollary \ref{the:the5}, as the level of stubbornness increases, the sample mean of steady-state opinions shifts from the expectation of message distribution or the sample mean of opinion leaders' steady-state opinions to the sample mean of initial opinions. We set $\alpha=\beta=1$, $a_m=b_x=2$, and $a_x=b_m=5$, and thus $\mu=\frac{1}{2}$, $\bar{m}^{(0)}=\frac{2}{7}$, and $\bar{x}^{(0)}=\frac{5}{7}$. Fig. \ref{fig:fig3c} shows the relationship between the sample mean of opinion leaders' steady-state opinions $\bar{\tilde{m}}$ and the level of stubbornness $\sigma$. We can see that as $\sigma$ increases, $\bar{\tilde{m}}$ shifts linearly from $\mu$ to $\bar{m}^{(0)}$. Fig. \ref{fig:fig3d} shows the relationship between the sample mean of normal agents' steady-state opinions $\bar{\tilde{x}}$ and the level of stubbornness $\rho$. Here, we choose $\sigma=\frac{1}{2}$ and $\theta=\frac{1-\rho}{2}$, and thus $\bar{\tilde{m}}=\frac{11}{28}$. We can see that as $\rho$ increases, $\bar{\tilde{x}}$ shifts from $\bar{\tilde{m}}$ to $\bar{x}^{(0)}$, and the curve corresponds to a fractional function. The above results are consistent with Corollary \ref{the:the5}.

\subsubsection{Sample variance of steady-state opinions}

\begin{figure*}[!t]
\centering
\subfloat[Initial opinion]{\includegraphics[width=1.75in]{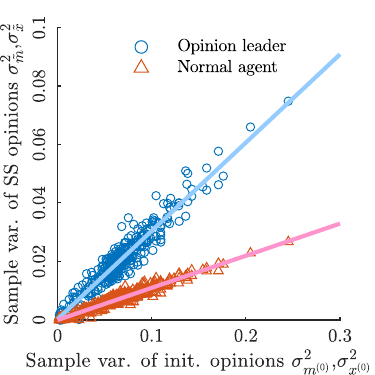} \label{fig:fig4a}}
\subfloat[Level of stubbornness]{\includegraphics[width=1.75in]{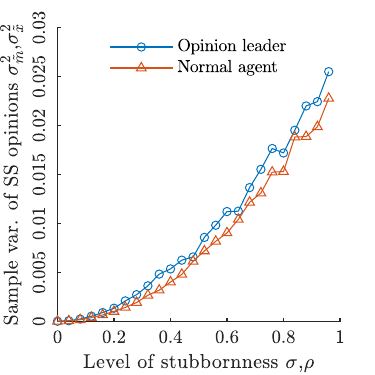} \label{fig:fig4b}}
\subfloat[Preference coefficient $\beta$]{\includegraphics[width=1.75in]{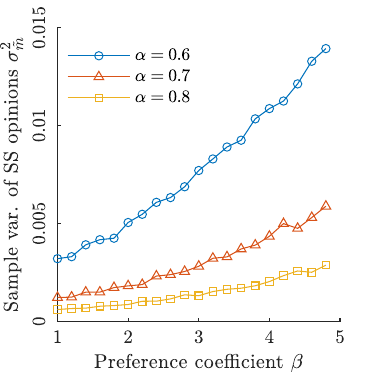} \label{fig:fig4c}}
\subfloat[Preference coefficient $\alpha$]{\includegraphics[width=1.75in]{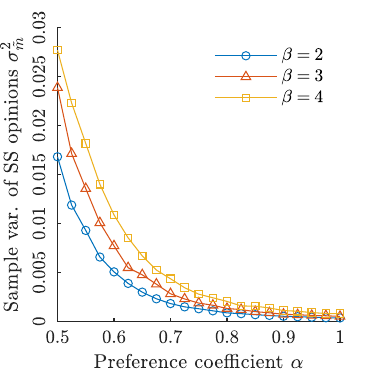} \label{fig:fig4d}}
\caption{Numerical experiment results: the sample variance of steady-state opinions.}
\label{fig:fig4}
\end{figure*}

Three factors influence the sample variance of steady-state opinions: initial opinion, level of stubbornness, and preference coefficient.

\textbf{Initial opinion}. We take the parameters $(a_m,b_m)$ throughout the interval $(0,2)\times(0,2)$ to change the sample variance of opinion leaders' initial opinions $\sigma^2_{m^{(0)}}$, which is equal to $\frac{a_mb_m}{(a_m+b_m)^2(a_m+b_m+1)}$. And we take the parameters $(a_x,b_x)$ throughout the interval $(0,2)\times(0,2)$ to change the sample variance of normal agents' initial opinions $\sigma^2_{x^{(0)}}$, which is equal to $\frac{a_xb_x}{(a_x+b_x)^2(a_x+b_x+1)}$. The relationship between the sample variance of steady-state opinions $\sigma^2_{\tilde{m}}$, $\sigma^2_{\tilde{x}}$ and the sample variance of initial opinions $\sigma^2_{m^{(0)}}$ and $\sigma^2_{x^{(0)}}$ is shown in Fig. \ref{fig:fig4a}. And we can see that ceteris paribus, the sample variance of steady-state opinions increases linearly with the sample variance of initial opinions, which is consistent with Corollary \ref{the:the7}.

\textbf{Level of stubbornness}. We also set $\alpha=\beta=1$, $a_m=b_x=2$, and $a_x=b_m=5$. The relationship between the sample variance of steady-state opinions $\sigma^2_{\tilde{m}}$, $\sigma^2_{\tilde{x}}$ and the level of stubbornness $\sigma$ and $\rho$ is shown in Fig. \ref{fig:fig4b}. And we can see that ceteris paribus, the sample variance of steady-state opinions increases with the level of stubbornness as a power or quadratic function, which is consistent with Corollary \ref{the:the5}.

\textbf{Preference coefficient}. First, we study the impact of $\beta$. We take $\beta$ throughout the interval $(1,5)$, and plot the relationship between the sample variance of steady-state opinions $\sigma^2_{\tilde{m}}$ and preference coefficient $\beta$ in Fig. \ref{fig:fig4c}. And we can see that ceteris paribus, the sample variance of steady-state opinions increases with $\beta$. Then, we study the impact of $\alpha$. We take $\alpha$ throughout the interval $(0.5,1]$, and plot the relationship between the sample variance of steady-state opinions $\sigma^2_{\tilde{m}}$ and preference coefficient $\alpha$ in Fig. \ref{fig:fig4d}. And we can see that ceteris paribus, the sample variance of steady-state opinions decreases with $\alpha$. The above results are consistent with Corollary \ref{the:the6}.

While we only present results for specific parameter values, the experiments can be repeated for other parameter configurations as well. Overall, the results provide evidence that our Two-Step Model is valid.

\section{Social Experiments}
\label{sec:social}
In this section, we conduct an opinion game social experiment to test the validity of the Two-Step Model and compare it with the baseline models.

\subsection{Experiment Design and Data Description}
To validate the opinion model, it is necessary to gather the opinion values of individuals in real-life situations over a period of time. Following prior work \cite{friedkin1990social, cartwright1971risk}, we propose an extension to the experiment, which includes a two-step process involving opinion leaders and normal agents, as well as the messages transmitted by the sources to the social networks.


Our experiment involves 2 tasks and requires the participation of 4 subjects in each task. Task 1 comprises 2 opinion games, where 4 agents complete the games individually. Task 2 consists of 7 opinion games, in which either 0, 1, or 2 subjects act as opinion leaders while the remaining subjects act as normal agents.
Initially, all 4 subjects read a paragraph describing a risk-taking scenario and form their opinions. Next, the opinion leader receives 3 or 4 messages, not available to the normal agents, and then expresses their opinion. Finally, the normal agents receive the opinion leader's opinion and revise their initial opinions. After that, all subjects assess the factors that influenced their opinion formation.

The on-site social experiment involved 120 subjects divided into 30 groups. In Task 1, each subject completed the decision-making process for 2 scenarios, resulting in 116 valid data points used to estimate the preference coefficients. In Task 2, all subjects within each group jointly completed the decision-making process for 7 scenarios, resulting in a total of 420 observed values, including 120 from the opinion leaders and 300 from the normal agents. The experiment investigated various variables for each subject, such as their initial and final (steady-state) opinions, level of stubbornness, influence matrix, and message selective coefficient in different scenarios.

\subsection{Result Analysis}
\subsubsection{Message preference}
The scenarios presented in Task 1 are not commonly encountered in daily life, and the subjects may have little prior knowledge about them. Therefore, their initial opinions may have little or no effect. This suggests that the level of stubbornness among the subjects can be considered negligible, and thus their initial opinions are ignored in Task 1.
To estimate the subjects' overall preference coefficients, which represent their inclination towards the arguments presented, we can use regression analysis based on the theory of statistics. The supplementary materials provide more information about the regression methods used in this study. According to the regression results, the estimated overall preference coefficients are $\hat{\alpha}=0.808<1$ and $\hat{\beta}=2.117>1$. The results of regression analysis for different scenarios also support this conclusion.


\subsubsection{Model accuracy}
We also investigate the correlation between the predicted opinions $\boldsymbol{m}^{\mathrm{p}}, \boldsymbol{x}^{\mathrm{p}}$ and observed values $\boldsymbol{m}^{\mathrm{o}}, \boldsymbol{x}^{\mathrm{o}}$. As previously stated, when $n\geqslant3$, the correlation between the simulation and theoretical values exceeds $0.96$, indicating that the number of messages is sufficient to validate the theorems. The linear regression result indicates that for both opinion leaders and normal agents, the slopes and intercepts are close to $1$ and $0$, respectively, which is presented in the supplementary materials. Consequently, the predicted value is a good approximation of the observed value.

\subsubsection{Comparison with other models}
\begin{table*}
\caption{RMSE between the predicted and observed steady-state opinions of opinion leaders and normal agents.}
\label{tab:tab4}
\centering
\setlength{\tabcolsep}{5.6pt}
\begin{tabular}{c|ccccccc|ccccccc}
\toprule
\multirow{1}{*}[-2pt]{Opinion model}&\multicolumn{7}{c|}{Opinion leaders}&\multicolumn{7}{c}{Normal agents}\\
\cmidrule{2-15}
\multirow{1}{*}[2pt]{of opinion leaders}&HK&BOF&SBC&LCSN&LnCSN&SinCSN&\textbf{MP}&HK&BOF&SBC&LCSN&LnCSN&SinCSN&\textbf{MP}\\
\midrule
Basketball game&0.273&0.277&\textbf{0.220}&0.226&0.225&0.227&0.223&0.167&0.168&0.144&\textbf{0.083}&0.165&0.161&0.154\\
Surgery&0.143&\textbf{0.135}&0.167&0.162&0.172&0.155&0.137&0.115&\textbf{0.113}&0.126&0.162&0.133&0.128&0.122\\
Business partner&0.189&0.187&0.221&0.188&0.197&\textbf{0.183}&0.199&0.150&0.148&0.159&\textbf{0.130}&0.135&0.136&0.150\\
College admission&\textbf{0.115}&0.120&0.179&0.149&0.168&0.135&0.117&\textbf{0.080}&0.081&0.103&0.099&0.105&0.095&0.084\\
Poll&0.206&0.188&0.205&0.174&0.175&0.174&\textbf{0.169}&\textbf{0.090}&0.094&0.105&0.097&0.099&0.096&0.091\\
Tourist attraction&\textbf{0.155}&0.192&0.175&0.163&0.174&0.156&0.161&\textbf{0.100}&0.120&0.113&0.114&0.118&0.110&0.103\\
\midrule
Overall&0.186&0.190&0.194&0.176&0.183&0.172&\textbf{0.170}&0.117&0.119&0.121&0.120&0.122&0.118&\textbf{0.116}\\
\bottomrule
\end{tabular}
\end{table*}

We also compare our MP model with benchmark models by analyzing the accuracy in predicting the steady-state opinions of both opinion leaders and normal agents, including the HK model \cite{hegselmann2002opinion}, BOF model \cite{dandekar2013biased}, SBC model \cite{liu2013social, baumann2021emergence}, LCSN model (CSN model with linear function), LnCSN model (CSN model with logarithmic function), SinCSN model (CSN model with sine function) \cite{mao2018spread}. And all parameters in these models are optimal and are shown in the supplementary materials. We evaluate the accuracy of the models by computing the root mean square error (RMSE) between the predicted steady-state opinions and the observed values. For example, the RMSE between the predicted steady-state opinions $\boldsymbol{m}^{\mathrm{p}}$ and the observed values $\boldsymbol{m}^{\mathrm{o}}$ of opinion leaders is given by
\begin{equation*}
    \textstyle
    \mathrm{RMSE}(\boldsymbol{m}^{\mathrm{p}}, \boldsymbol{m}^{\mathrm{o}})=\sqrt{\sum_{i=1}^p{\left(m^\mathrm{p}_i-m^\mathrm{o}_i\right)^2}\big/p}.
\end{equation*}
If the RMSE is lower, the model is considered more accurate. 
Table \ref{tab:tab4} demonstrates that our MP model outperforms the other models or exhibits minimal deviation from the best-performing model in terms of RMSE for opinion leaders' opinions across all scenarios. Moreover, our MP model shows the best RMSE performance for normal agents' opinions in the overall data. Therefore, these results suggest the validity of our model.

\section{Conclusion}
\label{sec:conclusion}
To model the two-step process including message sources, opinion leaders, and normal agents, we propose the Two-step Model. And we propose the concept of message preference for the selective exposure feature. Our results show that agents' opinions are influenced by multiple factors, such as message distribution, initial opinion, level of stubbornness, and preference coefficient, and the sample mean of normal agents' opinions tends to be influenced by opinion leaders. We conduct numerical and social experiments to verify the accuracy of our model, which demonstrates its effectiveness in outperforming existing benchmarks and better describing the two-step process.
Although our model considers only the two-step process, it provides a valuable contribution to the analysis of the multi-step process and lays a foundation for future research in this area.

\bibliographystyle{IEEEtran}
\bibliography{bibliology}


\vspace{11pt}

\end{document}